\documentclass[journal,twoside,web]{ieeecolor}
\usepackage{lcsys}
\usepackage{cite}
\usepackage{amsmath,amssymb,amsfonts}
\usepackage{algorithmic}
\usepackage{graphicx}
\usepackage{textcomp}
\def\BibTeX{{\rm B\kern-.05em{\sc i\kern-.025em b}\kern-.08em
    T\kern-.1667em\lower.7ex\hbox{E}\kern-.125emX}}

\usepackage{amsmath} 
\usepackage{amssymb}  
\usepackage{bm}
\usepackage{graphicx}
\usepackage{mathtools}
\usepackage{comment}
\usepackage[normalem]{ulem}
\usepackage{hyperref} 
\DeclareMathOperator*{\argmin}{arg\,min}

\usepackage[capitalise]{cleveref}

\newtheorem{assumption}{Assumption}
\newtheorem{theorem}{Theorem}
\newtheorem{definition}{Definition}
\newtheorem{lemma}{Lemma}

\newtheorem{remark}{Remark}

\crefname{assumption}{Assumption}{Assumptions}
\crefname{definition}{Definition}{Definitions}

\begin{document}
\title{Learning-based Prescribed-Time Safety for Control of Unknown Systems with Control Barrier Functions}
\author{Tzu-Yuan Huang$^{1}$, Sihua Zhang$^{2}$, Xiaobing Dai$^{1}$, Alexandre Capone$^{1}$, Velimir Todorovski$^{1}$, \\ Stefan Sosnowski$^{1}$ and Sandra Hirche$^{1}$, \IEEEmembership{Fellow, IEEE}
\thanks{This work was supported by the European Union’s Horizon Europe innovation action programme under grant agreement No. 101093822, "SeaClear2.0", by the DAAD programme Konrad Zuse Schools of Excellence in Artificial Intelligence, and by the BMBF with “Souverän. Digital. Vernetzt.” under 6G-life: 16KISK002.}
\thanks{$^{1}$The authors are with the Chair of Information-oriented Control, TUM School of Computation, Information and Technology, Technical University of Munich, Germany. {\tt\small $\{$tzu-yuan.huang, xiaobing.dai, alexandre.capone, velimir.todorovski, sosnowski, hirche$\}$@tum.de}. }
\thanks{$^{2}$The author is with the School of Automation, Beijing Institute of Technology, Beijing, People’s Republic of China.
        {\tt\small sihua.zhang@bit.edu.cn}.}}
\pagestyle{empty} 
\maketitle
\thispagestyle{empty} 
\begin{abstract}
In many control system applications, state constraint satisfaction needs to be guaranteed within a prescribed time. While this issue has been partially addressed for systems with known dynamics, it remains largely unaddressed for systems with unknown dynamics. In this paper, we propose a Gaussian process-based time-varying control method that leverages backstepping and control barrier functions to achieve safety requirements within prescribed time windows for control affine systems. It can be used to keep a system within a safe region or to make it return to a safe region within a limited time window. These properties are cemented by rigorous theoretical results. The effectiveness of the proposed controller is demonstrated in a simulation of a robotic manipulator.
\end{abstract}


\section{INTRODUCTION}

\IEEEPARstart{C}{ontrol} systems with active constraints during limited time windows are ubiquitous. For example, in a robot-human handover scenario \cite{HandoverTask}, contact constraints are only relevant as long as the object is being handed over. This interaction temporarily modifies the safety region, reflecting the dynamic nature of the environment. The time-limited nature of the constraints in such settings allows for considerable flexibility, which can be leveraged to improve control performance. However, most existing algorithms aim to enforce safety at all times based on the initial safe condition, as opposed to relaxing these requirements when permissible, yielding potentially overly conservative behavior. Recent works have addressed a less stringent notion of safety than typically found in the literature, where safety constraints are only considered for a finite time window \cite{dai_adaptive_2022}.
The goal is to guarantee that the system returns and remains in the safe region within a pre-specified time interval. While this type of task has been addressed in settings with known dynamics, the considerably more challenging problem with unknown dynamics remains largely unaddressed.

Recently, several techniques addressing the issue of enforcing the system in the safe region have been proposed within the temporal constraints. Model predictive control \cite{2020_MPC2} is widely adopted in dynamic systems to achieve optimal performance while satisfying multiple constraints. However, this approach demands substantial computational resources since it requires solving a sequence of constrained optimization problems within a finite time horizon at each discrete time step. Safe reinforcement learning \cite{safe_RL1} is another powerful tool to address safety issues even under an uncertain environment. Nevertheless, its practical deployment is challenging due to the large sim-to-real gaps and theoretical guarantee of safety. Control barrier functions (CBFs) \cite{ames_control_2014} are increasingly utilized to ensure safety in systems, employing a quadratic program (QP) with linear constraints at every discrete time step. Benefiting by the QP framework, CBFs method could be applied as a real-time optimization-based controller, and the safety is guaranteed by rigorous proof. However, vanilla CBF lacks the consideration of temporal constraints. In settings where the system has to be within a safe region after a pre-specified time, the prescribed-time safety (PTSf) controller is devised in \cite{abel_prescribed-time_2023} based on a CBF framework. With a design of time-varying gain in the CBFs, the safety of the system is guaranteed in the specified time horizon starting from an initially safe condition. Compared to other CBF frameworks with time requirements such as finite-time \cite{li_formally_2018} and fixed-time safety controller \cite{garg_multi-rate_2022}, the PTSf controller benefits from its simple design and independence of the initial system states. However, the design of the PTSf controller crucially relies on the availability of accurate system dynamics, restricting its practical usage in cases with unknown dynamics or environmental uncertainties. In addition, the capability of returning to the safe regions from the unsafe initial condition for PTSf controllers is not shown in previous work.
\looseness=-1

Supervised machine learning techniques are increasingly promising for identifying unknown dynamical systems from data. However, adequately accounting for model uncertainty remains an open problem for safe control. In \cite{taylor_learning_nodate}, a neural network model is used to estimate model uncertainty, which is then leveraged together with robust  CBFs to guarantee the safety of the closed-loop system. The work of \cite{yaghoubi_training_2020} proposes a CBF-based imitation learning approach, where a deep neural network mimics the outcome of CBF-based controllers. In \cite{jagtap_control_2020}, 
a Gaussian process (GP) model is employed to formulate a robust CBF, which is then leveraged to derive a control law that renders the system safe. The works in \cite{lederer_safe_2023} and \cite{zhang2023safety} employ GP regression to model an elastic-joint robot, which is then rendered safe using a robust CBF. However, none of the techniques mentioned above address the initial unsafe condition that results from the temporary alteration of the safe region. Instead, these works always maintain the safety from a safe state at all times.

In this paper, we consider an unknown high-order control affine system with controllable canonical form, which requires to stay in a pre-defined safe set within a prescribed time. To this end, We propose a novel, robust Gaussian process-based framework for prescribed time control barrier functions incorporating the probabilistic uncertainty quantification. Our rigorous proof verifies that the system reliably remains in or returns to the safe set in a pre-defined time horizon with a high probability, regardless of whether the initial system state is safe or unsafe. The effectiveness of our method is demonstrated using a numerical simulation of a two-link robot manipulator.

The remainder of this paper is structured as follows. In \cref{section_problem_setting}, the system setting and the safety requirements are introduced. In \cref{section_GP_PT_CBF}, a time-varying Gaussian process-based control barrier function framework is proposed with mathematical proof. The effectiveness of the proposed method is shown via simulation for a two-link robot manipulator in \cref{section_simulation}, followed by the conclusion in \cref{section_conclusion}.
\section{Problem Setting and Preliminaries} \label{section_problem_setting}
\subsection{System Description}
Consider a nonlinear continuous-time system with unknown dynamics in the controllable canonical form as
\begin{equation}
\begin{aligned} \label{eq:1}
	&\dot{\bm{x}}_i  =\bm{x}_{i+1}, \qquad i=1,...,n-1, \\
	&\dot{\bm{x}}_n  = \bm{f}(\bm{x})+\bm{g}({\bm{x}})\bm{u} + \bm{d}(\bm{x}), 
\end{aligned}
\end{equation}
where $\bm{x}= [\bm{x}_1^\top, \cdots, \bm{x}_n^\top ]^\top \in \mathbb{X} \subset \mathbb{R}^{mn}$ with $\bm{x}_i = \left[ x_{i,1},\cdots,x_{i,m} \right]^\top \in \mathbb{R}^m$, $\forall i = 1,\cdots, n$ represents the system state in the compact domain $\mathbb{X}$ and $\bm{u} \in \mathbb{U} \subseteq \mathbb{R}^m$ denotes the control input at time $t \in \mathbb{R}_{+}$.
The functions $\bm{f}:\mathbb{X} \rightarrow  \mathbb{R}^{m}$ and $\bm{g}:\mathbb{X} \rightarrow  \mathbb{R}^{m\times m}$ are locally Lipschitz continuous functions that represent the known components of the system dynamics, the function $\bm{d}: \mathbb{X} \rightarrow  \mathbb{R}^{m}$ in \eqref{eq:1} encodes all state-dependent model uncertainties from, e.g., environmental effects and unmodeled parts of the system dynamics. This form is a common structure in many practical systems such as robot manipulators \cite{huang2023rapid}.
Moreover, we make the following assumption regarding $\bm{g}$.
\begin{assumption}\label{a_g_nonsigular} 
	For all $\bm{x} \!\in\! \mathbb{X}$, $\bm{g}({\bm{x}})$  is non-singular. 
\end{assumption}
\cref{a_g_nonsigular} is reasonable for various types of systems, e.g., manipulators, and is frequently satisfied by control-affine systems \cite{slotine1993robust}. It implies that we can generate control inputs to compensate for nonlinearities in an arbitrary direction.
\looseness=-1

In order to design a control law that ensures safety, we require an adequate model of the uncertainty $\bm{d}$. The data-driven model is employed to infer the uncertainty, where a noisy measurement data $\mathbb{D}$ specified by the following assumption is leveraged.

\begin{assumption} \label{a_DataSet}
	The data set $\mathbb{D}$ consists of $N \!\in\! \mathbb{N}$ training pairs $\{ \bm{x}^{(k)}, \bm{y}^{(k)} \}$ with $\bm{y}^{(k)} \!=\! \bm{d}(\bm{x}^{(k)}) \!+\! \bm{\epsilon}^{(k)}, \forall k \!=\! 1, \!\cdots\!, N$, where $\bm{\epsilon}^{(k)}$ is i.i.d. zero-mean Gaussian noise with covariance $\bm{\Sigma}_o \!=\! \mathrm{diag}(\sigma_{o,\!1}^2, \!\cdots\!,\sigma_{o,\!m}^2)$, $\sigma_{o,j} \!\in\! \mathbb{R}_{+}, \forall j \!=\! 1,\cdots,m$.
\end{assumption}

\cref{a_DataSet} is a mild assumption often encountered in learning-based control settings \cite{capone2019backstepping}. It allows for Gaussian distributed measurement noise, which can be due to, e.g., numerical differentiation. In some settings, the requirements for the measurement noise distribution can be relaxed, e.g., by restricting it to be bounded \cite{hashimoto2022learning}, but this is out of the scope of this paper.

\subsection{Prescribed-Time Safety (PTSf)}
In this paper, a safe set $\mathcal{C} \!\subseteq\! \mathbb{X}$ is defined by a known, continuously differentiable \textit{control barrier function} (CBF) $h(\bm{x}): \mathbb{X} \!\rightarrow\! \mathbb{R}$ as $\mathcal{C} \!=\! \{\bm{x}\in \mathbb{X}: h(\bm{x})\geq 0\}$. Specifically, the system is considered safe if $\bm{x} \!\in\! \mathcal{C}$, and unsafe otherwise.

If the system \eqref{eq:1} is known perfectly, then the CBF $h$ can be leveraged to compute certifiably safe control inputs
\cite{ames_control_2014,Aron_2017}. 
For many practical systems, it is possible to derive an appropriate CBF with only imperfect knowledge of the system at hand, e.g., adaptive cruise control system \cite{castaneda2022probabilistic}.

We now introduce the notion of \textit{prescribed-time} safety, which is the main focus of this paper. Based on~\cite{abel_prescribed-time_2023}, we distinguish between PTSf for systems that are initially safe and unsafe. The latter case corresponds to rescuing safety within a prescribed time~\cite{abel_prescribed-time_2023}.

\begin{definition}[PTSf for initially safe system]\label{d1}
Consider the system \eqref{a_g_nonsigular}. If the initial state is safe, i.e., $\bm{x}(t_0)\in \mathcal{C}$, then the system is said to be PTSf with a prescribed time $T_{\text{pre}} \in \mathbb{R}_{+}$ if $h(\bm{x}(t)) \!\geq\! 0, \forall \ t \!\in\! [t_0, t_0 \!+\! T_{\text{pre}})$. \looseness=-1
    \end{definition}

\begin{definition}[PTSf for initially unsafe system]\label{d2}
Consider the system \eqref{a_g_nonsigular}. If the initial state is unsafe, i.e., $\bm{x}(t_0)\notin \mathcal{C}$, then the system is said to be PTSf with a prescribed time $T_{\text{pre}} \in \mathbb{R}_{+}$ if $h(\bm{x}(t_0+T_{\text{pre}})) \geq 0.$
\end{definition}

Our goal is then to design a control algorithm that guarantees PTSf for the system \eqref{eq:1} whenever the initial state is either safe or unsafe.
\section{Learning-based Control with Prescribed-time Safety}\label{section_GP_PT_CBF}
To address the PTSf problem, Gaussian process regression is adopted as a data-driven approach to approximate the uncertainty $\bm{d}(\bm{x})$ in \eqref{eq:1}. Based on GP regression, we propose a learning-based PTSf controller to guarantee the safety objectives defined in \cref{d1} and \cref{d2} with system uncertainty.

\subsection{Gaussian Process Regression}
Gaussian process regression, as a non-parametric method, is widely used to approximate unknown continuous functions due to its modeling flexibility.
In order to learn the $m$-dimensional unknown function $\bm{d}(\cdot) = [d_1(\cdot), \cdots, d_m(\cdot)]^\top$ from data set $\mathbb{D}$ satisfying \cref{a_DataSet}, each component $d_j$ is represented as a GP $d_j \sim \mathcal{GP}(m_j(\cdot),k_j(\cdot,\cdot)), \forall j = 1,\cdots,m$, which is specified by the prior mean $m_j(\cdot):\mathbb{X} \rightarrow \mathbb{R}$ and Lipschitz covariance function $k_j(\cdot,\cdot):\mathbb{X} \times \mathbb{X} \rightarrow \mathbb{R}_{0,+}$. The mean function $m_j(\cdot)$ encodes the prior knowledge of the system, which in our case is included in $f(\cdot)$ resulting in $m_j(\bm{x}) = 0$ for $\forall \bm{x} \in \mathbb{X}$.
The covariance function, also called kernel function, resulting in $k_j(\bm{x}, \bm{x}')$ reflects the correlation between evaluations of $d_j(\cdot)$ at state $\bm{x}$ and $\bm{x}'$ with $\bm{x}, \bm{x}' \in \mathbb{X}$. The unknown function $\bm{d}(\bm{x})$ now is expressed as $m$ GP models
	\begin{equation}
		\begin{aligned}
			\bm{d}(\bm{x}) = \left\{ \begin{array}{c}
				d_1 \sim \mathcal{GP}(0,\kappa_1(\cdot,\cdot)) \\
				\vdots \\[5pt] 
				d_m \sim \mathcal{GP}(0,\kappa_m(\cdot,\cdot))\\
			\end{array}\right.
		\end{aligned}
	\end{equation}

Utilizing the data set $\mathbb{D}$ satisfying \cref{a_DataSet} with $| \mathbb{D} | = N$ and Bayesian principle, the value of $\bm{d}(\bm{x})$ follows a Gaussian distribution characterized by the posterior mean $\boldsymbol{\mu}(\bm{x})=[ \mu_1(\bm{x}), \cdots, \mu_m(\bm{x}) ]^{\top}$ and variance $\bm{\Sigma}(\bm{x})=\mathrm{diag}\left( \sigma^2_1(\bm{x}), \cdots, \sigma^2_m(\bm{x}) \right)$ with
\begin{align}
	&\mu_j(\bm{x}) = \bm{k}_{\bm{X}_j}^{\top}(\bm{x})(\bm{K}_j+\sigma_{n}^2\bm{I})^{-1} \bm{y}_j, \label{eq:GP_Mean}
	\\
	&\sigma_j^2(\bm{x})=\kappa_j(\bm{x},\bm{x})-\bm{k}_{\bm{X}_j}^{\top}(\bm{x})(\bm{K}_j+\sigma_{n}^2\bm{I})^{-1} \bm{k}_{\bm{X}_j}(\bm{x}), \label{eq:GP_PosteriorVariance}
\end{align}
where the kernel vector and gram matrix are specified as $\bm{k}_{\bm{X}_j}(\bm{x})=[\kappa_j(\bm{x}^{(1)},\bm{x}),\cdots,\kappa_j(\bm{x}^{(N)},\bm{x})]^{\top}$ and $\bm{K}_j=[\kappa_j(\bm{x}^{(i)},\bm{x}^{(k)})]_{i,k=1,\cdots,N}$ , respectively. The vector $\bm{y}_j=[y_j^{(1)},\cdots, y_j^{(N)}]^{\top}$ with $y_j^{(k)}$ representing the $j^{th}$ dimension of $\bm{y}^{(k)}$ is the concatenation of the output values in data set $\mathbb{D}$.
The posterior mean function $\bm{\mu}(\cdot)$ serves as a prediction model of the unknown function $\bm{d}(\cdot)$, whereas the variance $\bm{\Sigma}(\cdot)$ is employed as an indicator of epistemic uncertainty, which is shown as follows.

\begin{lemma}[\cite{lederer_uniform_nodate}]\label{l1}
	Consider an unknown function $d_j(\cdot)$ for $\forall j = 1, \cdots, m$ and a data set satisfying \cref{a_DataSet}.
	Choose $\tau \in \mathbb{R}_+$ and $\delta \in (0,1) \subset \mathbb{R}$, then
	\begin{equation}
		\begin{aligned}
			&\Pr\left\{|d_j(\bm{x})-\mu_j(\bm{x})|\leq \eta_j(\bm{x}), \forall \bm{x} \in \mathbb{X}\right\} \geq 1-\delta,
			\\ 
			&\eta_j(\bm{x}) = \sqrt{\beta_{\delta}(\tau)}\sigma_j(\bm{x}) + \gamma_{\delta}(\tau), \label{eq:eta}
		\end{aligned}
	\end{equation}
	where $\gamma_{\delta}(\tau)=(L_{d,j}+\sqrt{\beta_{\delta}(\tau)}L_{\sigma,j}+L_{\mu,j})\tau$ and
	\begin{align}
		& \beta_{\delta}(\tau) = 2\sum\nolimits_{j=1}^{mn}\log \Big( \dfrac{0.5\sqrt{mn}(\bar{x}_j-\underline{x}_j)}{\tau\delta} +\dfrac{1}{\delta}\Big),
	\end{align}
	and $\bar{x}_j = \max_{\bm{x} \in \mathbb{X}}x_j,\underline{x}_j=\min_{\bm{x} \in \mathbb{X}}x_j$ with $x_j$ refering to the $j$-th dimension of $\bm{x}$.
	The constants $L_{\mu,j}$, $L_{\sigma,j}$, $L_{d,j} \in \mathbb{R}_+$ are the Lipschitz constants of mean $\mu_j(\cdot)$, standard deviation $\sigma_j(\cdot)$ and function $d_j(\cdot)$, respectively.
\end{lemma}

\cref{l1} provides a probabilistic bound for the prediction error from GP regression, which is widely used in safety-critical applications. The detailed expressions of $L_{\mu,j}$ and $L_{\sigma,j}$ can be found in \cite{lederer_uniform_nodate}, and the Lipschitz constant $L_{d,j}$ for the unknown function $d_j(\cdot)$ can be approximated as in \cite{lederer_uniform_nodate}. \looseness=-1

\subsection{PTSf Design with Learning Uncertainty}
Recalling the system as described in (\ref{eq:1}), the aim of this subsection is to develop a safety controller for the control input, ensuring that (\ref{eq:1}) satisfies \cref{d1} and \ref{d2} with a $n^{th}\!$ order differentiable CBF $\bm{h}(\bm{x}_1)\!:\! \mathbb{R}^m \!\to\! \mathbb{R}^{d_h}$ with $d_h \!\in\! \mathbb{N}_+$. Without loss of generality, we consider the scalar CBF with $d_h \!=\! 1$, whose result can be extended to high dimensional $\bm{h}(\cdot)$.

The design of the prescribed-time safety controller incorporates a blow-up function, which is defined as follows:
\begin{align} \label{eq:blowupfunction}
	\varphi(t)=\dfrac{T_{\text{pre}}^2+\alpha((t-t_0)^2-(t-t_0)T_{\text{pre}})^2}{(T_{\text{pre}}+t_0-t)^2}, && t \geq t_0,  
\end{align}
where $\alpha \in \mathbb{R}_{0,+}$ represents a scalar tunable parameter to the convergence speed.
Note that $\varphi(\cdot)$ is an increasing positive function in the time horizon $[t_0,t_0 \!+\! T_{\text{pre}})$, which provides the flexibility to achieve PTSf in our following design.
With the definition of the time-varying function $\varphi(\cdot)$ in \eqref{eq:blowupfunction}, a series of barrier functions are designed as
\begin{align}
	&h_1(\bm{x}_{1}) = h(\bm{x}_1), \label{eq:h1}\\
	&h_{i+1}(t, \bm{x}_{1:i+1}) = \dot{h}_i(t, \bm{x}_{1:i}) + c_i \varphi(t) h_i(t, \bm{x}_{1:i}), \label{eq:hi+1}	
\end{align}
for $i=1,\cdots,n - 1$, where $\bm{x}_{1:i} = [\bm{x}_1^\top, \cdots, \bm{x}_i^\top]^\top$ and $c_i \in \mathbb{R}_+$ are positive constants to be determined later.
The time derivatives $\dot{h}_i$ for $i = 1, \cdots, n$ are explicitly written according to \eqref{eq:1} as
\begin{align}
    &\dot{h}_i(t, \bm{x}_{1:i}) = \sum\nolimits_{j=1}^{i} \frac{\partial h_i(t, \bm{x}_{1:i})}{\partial \bm{x}_j} \bm{x}_{j+1} + \frac{\partial h_i(t, \bm{x}_{1:i})}{\partial t} \label{eq:hi_dot}
\end{align}
for $i=1,\cdots,n-1$ and
\begin{align} \label{eq:h_n+1}
    h_{n+1}&(t, \bm{x}, \bm{u}) = \dot{h}_n(t, \bm{x}) + c_n \varphi(t) h_n(t, \bm{x}) \\
    =& \sum\nolimits_{j=1}^{n-1}\frac{\partial h_n(t, \bm{x})}{\partial \bm{x}_j} \bm{x}_{j+1} \!+\! \frac{\partial h_n}{\partial \bm{x}_n}( \bm{f}({\bm{x}}) \!+\! \bm{g}({\bm{x}})\bm{u} \!+\! \bm{d}(\bm{x}) )  \nonumber \\
    & + c_n \varphi(t) h_n(t, \bm{x}) \nonumber
\end{align}

Note that, the control input $\bm{u}$ and uncertainty of the system $\bm{d}(\cdot)$ are included in \eqref{eq:h_n+1}, making $h_{n+1}(t, \bm{x}, \bm{u})$ impossible to evaluate.
Instead, the posterior mean $\bm{\mu}(\cdot)$ and variance $\bm{\Sigma}(\cdot)$ obtained from Gaussian process in \eqref{eq:GP_Mean} and \eqref{eq:GP_PosteriorVariance} are employed to approximate $h_{n+1}(\bm{x}_{n+1},t)$.
For notational simplicity, we denote $h_i(t) \coloneqq h_i(t, \bm{x}_{1:i}(t))$ and $\dot{h}_i(t) \coloneqq \dot{h}_i(t, \bm{x}_{1:i}(t))$ and the control performance is shown in the following theorem.

\begin{theorem}\label{t1}
	Consider the system \eqref{eq:1} and let \cref{a_g_nonsigular} and \ref{a_DataSet} hold. 
    Let $\bm{\mu}(\cdot)$ and $\eta_j(\cdot), j = 1,\cdots, m$ be as in \eqref{eq:GP_Mean} and (\ref{eq:eta}) respectively and choose $\delta \in (0,1/m)$. 
    Let $\bm{u}_{\text{nom}}$ be the control input provided by other nominal controllers, e.g., PID controller, feedback linearization, and $\bm{u}_{\text{safe}}$ be obtained by solving the quadratic programming (QP) as
	\begin{subequations}\label{eq:main_theory}
		\begin{align}
			& \bm{u}_{\text{safe}} = \argmin\nolimits_{\bm{u} \in \mathbb{U}}  ||\bm{u} - \bm{u}_{nom}||^2 \\
			& \mathrm{s.t.} \quad h_{n+1}^*(t, \bm{x},\bm{u}) \geq 0,
		\end{align}
	\end{subequations}
    in which 
    \begin{align}\label{eq:theorem_term}
         &h_{n+1}^*(t, \!\bm{x},\!\bm{u}) \!=\! \sum\nolimits_{j\!=\!1}^{n\!-\!1}\dfrac{\partial h_n}{\partial \bm{x}_j}\bm{x}_{j\!+\!1} \!+\! \frac{\partial h_n}{\partial t} \!+\! c_n \varphi(t) h_n(t) \\
		&+ \dfrac{\partial h_n}{\partial \bm{x}_n}(\bm{f}(\bm{x}) + \bm{g}({\bm{x}})\bm{u} + \boldsymbol{\mu}(\bm{x}) )  - \sum\nolimits_{j=1}^m \left|\dfrac{\partial h_n }{\partial x_{n,j}}\right| \left|\eta_j(\bm{x})\right| \nonumber
    \end{align}
	with initial gains $c_i$ for $i = 1, \cdots, n$ satisfying $c_n > 0$ and
	\begin{equation}
		\begin{split}
			& c_i > \max \{0,- h_i^{-1}(t_0) \dot{h}_i(t_0) \}, i=1,\cdots,n-1,
		\end{split}
	\end{equation}
	If the QP \eqref{eq:main_theory} is feasible for all $\bm{x}\in \mathbb{X}$ and all $t \geq t_0$, then the control input $\bm{u} = \bm{u}_{\text{safe}}$ in \eqref{eq:1}
	guarantees PTSf according to \cref{d1} and \ref{d2} with probability of at least
	$1-m\delta$.
\end{theorem}

\begin{proof}
	Our proof is structured into two parts. In the first part, we will show that $h_{n+1}(t) \ge 0, \forall t \ge t_0$ is satisfied with a high probability. Consider the subtraction of $h_{n+1}^*(t)$ in \eqref{eq:theorem_term} from $h_{n+1}(t)$ in \eqref{eq:h_n+1}, which is written as
	\begin{align} \label{eq:psi}
			h_{n+1}(t) -& h_{n+1}^*(t) \\
			&=  \dfrac{\partial h_n}{\partial \bm{x}_n}(\bm{d}(\bm{x})-\bm{\mu}(\bm{x})) + \sum\nolimits_{j=1}^m \left|\dfrac{\partial h_n}{\partial x_{n,j}}\right| \left|\eta_j(\bm{x})\right|. \nonumber
	\end{align}
	Note that \eqref{eq:psi} is only related to $\bm{x}$, such that 
    by applying the uniform probabilistic error bound in \cref{l1} to \eqref{eq:psi} and the extension to $m$ dimension through Boole's inequality \cite{venkatesh2013theory, Union_bound}, the positivity of $h_{n+1}(t) - h_{n+1}^*(t)$ is guaranteed within a probability bound as 
	\begin{gather}
		\Pr\left\{ h_{n+1}(t) - h_{n+1}^*(t) \ge 0, \forall \bm{x} \in \mathbb{X} \right\} \geq 1-m\delta. \label{eq:Pr_h_minus_h*}
	\end{gather}
    By picking the values of $\delta \leq 1/m \in \mathbb{R}^+$, a high probability of \eqref{eq:Pr_h_minus_h*} and the following equation can be guaranteed.  
    
    \eqref{eq:Pr_h_minus_h*} implies that $h_{n+1}(t) \geq h_{n+1}^*(t)$ with a probability of at least $1-m\delta$ for $\forall t \geq t_0$ with $\bm{x}(t) \in \mathbb{X}$. 
    Considering that in the compact domain $\mathbb{X}$, $h_{n+1}^*(t) \geq 0$ for $\forall t \geq t_0$ is guaranteed in \eqref{eq:main_theory}, this consequently proves that
	\begin{equation}
        \begin{aligned}
            &\Pr\{h_{n+1}(t) \geq 0, \forall t \geq t_0\} \geq 1 - m \delta.  \label{eq:h_n+1>0withProb}\\
            &h_{n+1}\coloneqq h_{n+1}(t, \bm{x}_{n}(t)), \forall \bm{x}_{n}(t) \in \mathbb{X} 
        \end{aligned}
	\end{equation}
    A similar procedure is also employed in \cite[Th4.3]{jagtap_control_2020}. In the second part of our proof, we aim to guarantee the prescribed-time safety, regardless of whether the initial system condition is safe or unsafe. Based on \eqref{eq:h_n+1>0withProb}, it implies from \eqref{eq:h_n+1} that
	\begin{gather}
		\Pr\{\dot{h}_n(t) \geq -c_n \varphi(t)  h_n(t), \forall t \geq t_0 \}\geq 1-m \delta. \label{eq:dh_ndt>cmuhwithProb}
	\end{gather}
	By applying the variation of constants formula and the comparison lemma \cite{hassan2002nonlinear}, for the time horizon of $[t_0, t_0 + T_{\text{pre}})$, the solution of \eqref{eq:dh_ndt>cmuhwithProb} is derived as
	\begin{gather}
		\Pr\Big\{ h_n(t) \!\geq\! h_n(t_0) e^{-c_n\int\nolimits_{t_0}^t \varphi (s)ds}, \forall t \!\ge\! t_0 \Big\} \!\geq\! 1 \!-\! m\delta.  \label{eq:hn>inthn}
	\end{gather}
	Similarly, the analytical solutions of $h_i(t)$ for $i=1,\cdots,n-1$ in \eqref{eq:hi+1} are also reformulated as
	\begin{equation} \label{eq:hi_int}
		\begin{aligned} 
			h_{i}(t)&=\int\nolimits_{t_0}^te^{-c_{i}\int\nolimits_{\tau}^t\varphi (s)ds}h_{i+1}(\tau)d\tau \!+\!  h_{i}(t_0)e^{-c_{i}\int\nolimits_{t_0}^t\varphi (s)ds}.
		\end{aligned}
	\end{equation}
	Moreover, by substituting \eqref{eq:hn>inthn} into \eqref{eq:hi_int}, it leads to
	\begin{align} \label{eq:h_twoexp}
        h_{i}(t) \geq& h_{i+1}(t_0)\int\nolimits_{t_0}^t e^{-\left[c_{i}\int\nolimits_{\tau}^t\varphi (s)ds + c_{i+1}\int\nolimits_{t_0}^{\tau}\varphi (s)ds\right]} d\tau \\
        &+ h_{i}(t_0) e^{-c_{i}\int\nolimits_{t_0}^t\varphi (s)ds} \nonumber
	\end{align}
    with probability of at least $1 - m\delta$ for $i=n-1$.
    
	We now start to demonstrate that the prescribed-time safety for an initially unsafe system is guaranteed. A series of auxiliary gains $c_i^*$ for $\forall i = 1,\cdots,n$ are defined as
	\begin{gather} \label{eq:c*}
		c_i^* = \begin{cases}
			\Bar{c}, & \text{if $h_i(t_0) > 0$},   \\
			\underline{c}, & \text{otherwise} \\
		\end{cases},
	\end{gather}
	where $\bar{c}=\max\{c_1,\cdots,c_n\}$, and $\underline{c}=\min\{c_1,\cdots,c_n\}$. Due to the choice of $c_i$ in \cref{t1}, the auxiliary gains $c_i^*$ are non-negative. Next, we rewrite \eqref{eq:h_twoexp} for $i=n-1$ as
	\begin{align} \label{eq:h_i_auxiliary}
		h_{i}(t) \geq& h_{i+1}(t_0) \int\nolimits_{t_0}^te^{-c_i^*\int\nolimits_{\tau}^t\varphi (s)ds}d\tau + h_{i}(t_0)e^{-c_{i}^*\int\nolimits_{t_0}^t \varphi (s)ds}, 
	\end{align}
	which inherits the probability of at least $1-m\delta$. By applying the induction step with $i$ to \eqref{eq:hi_int} from $n-1$ to $1$ recursively, the inequality of $h_i(t)$ for $\forall i=1,\cdots,n-1$ is expressed as
	\begin{gather}
		h_{i}(t) \geq \sum\nolimits_{j=i}^n h_j(t_0) \dfrac{(t-t_0)^{j-i}}{(j - i)!} e^{c_{j}^* \int\nolimits_{t_0}^t \varphi (\tau)d\tau}.
	\end{gather}
	As a result, the value of $h_1(t)$ is bounded as
	\begin{align} \label{eqn_h1_t}
		h_1(t) \!\ge\! \sum\nolimits_{j\!=\!1}^n h_j(t_0) \frac{ (t \!-\! t_0)^{j \!-\! 1\!}}{(j \!-\! 1)} e^{-c_j^* \int\limits_{t_0}^{T^*}  \varphi (\tau)d\tau} e^{-c_j^* \int\limits_{T^*}^t  \varphi (\tau)d\tau},
	\end{align}
	where $T^* \!=\! t_0 \!+\! T_{\text{pre}}$. Note that with the definition of $\varphi (\cdot)$ in \eqref{eq:blowupfunction} and the positivity of $c_j^*$, it has $e^{-c_j^*\int\nolimits_{t_0}^{T^*}\varphi (\tau)d\tau} \!=\! 0$,
	such that $h_1(t) \!\ge\! 0, \forall t \!\ge\! t_0 \!+\! T_{\text{pre}}$ holds with a probability of at least $1 \!-\! m \delta$, which concludes the proof for an initial unsafe condition.
	We now prove that the prescribed-time safety for
	an initially safe system is also guaranteed. With $\varphi (t_0) \!=\! 1$ at the initial time, the equality in \eqref{eq:hi+1} for $i \!=\! 2,\!\cdots\! n$ is carried out as
    \looseness=-1
	\begin{gather}
		h_{i}(t_0)=\dot{h}_{i-1}(t_0)+c_{i-1}h_{i-1}(t_0).
	\end{gather}
	Through the design of initial gains in Theorem \ref{t1}, the initial value $h_{i}(t_0)$ follows that $h_{i}(t_0)>0, \quad i=2,\cdots n$.
	Additionally, leveraging the positive nature of the exponential integral, the inequality \eqref{eq:h_twoexp} is further written as
	\begin{gather}
		h_{i}(t) \geq h_{i}(t_0)e^{-c_{i}\int\nolimits_{t_0}^t\varphi (s)ds}>0, \ i=n-1, \label{eq:hi-1>exp}
	\end{gather}
	which inherits the probability of at least $1-m\delta$. By substituting \eqref{eq:hi-1>exp} to \eqref{eq:hi_int} for each steps backwards from $i=n-1$ to $1$, the subsequent formula can be expressed as
	\begin{gather}
		h_{1}(t) \geq h_{1}(t_0)e^{-c_1\int\nolimits_{t_0}^t\varphi (s)ds} \label{eq:h1_greater_exp}
	\end{gather}
	Consider a system that initially remains within the safe area, i.e., $h_1(t_0)\geq 0$, and proceeding from (\ref{eq:h1_greater_exp}), it is shown that
	\begin{gather}
		h_{1}(t) \geq h_{1}(t_0)e^{-c_1\int\nolimits_{t_0}^t\varphi (s)ds} \geq 0, \forall t \in [t_0,t_0+T_{\text{pre}}) \label{eq:h1_safetosafe}
	\end{gather}
	with at least $1 \!-\! m\delta$ probability. Combining the proof for initially safe and unsafe cases, it is proven that the prescribed-time safety is guaranteed with a high probability. \looseness=-1
\end{proof}

\cref{t1} shows the prescribed time safety in \cref{d1} and \ref{d2} is achieved with high probability by using the proposed GP-based safety controller in \eqref{eq:main_theory}, relaxing the requirement of the known accurate model as in \cite{abel_prescribed-time_2023}.
Despite probabilistic safety, the proposed controller only requires Lipschitz continuity of $\bm{d}(\cdot)$, which is common in nonlinear control \cite{hassan2002nonlinear} and less restrictive than other methods based on e.g., neural networks \cite{yaghoubi_training_2020}. To guarantee the safety with higher probability, a more conservative approximation of the prediction error is non-negligible according to \cref{l1}, inducing larger $\eta_j(\cdot), j = 1,\cdots,m$ and causing potential infeasibility of the QP problem in \eqref{eq:main_theory}. To reduce the prediction errors and improve the feasibility, the incorporation of distributed GP \cite{lederer2021gaussian} and online learning \cite{capone2023safe} is an efficient and promising way, which can be directly integrated into our framework.

\begin{remark}
    The barrier functions $h_i$, $i \!=\! 1,\!\cdots\!, n\!+\!1$ in \eqref{eq:h1}-\eqref{eq:hi_dot}, can be extended to a multi-dimensional function with $d_h \!>\! 1$, which represents multiple safety constraints. With respect to the extension of barrier functions, system safety is proven with a similar process from \eqref{eq:psi} to \eqref{eq:h1_safetosafe}.\looseness=-1
\end{remark}

\begin{remark}
    In this paper, the proposed method only guarantees the PTSf for the system \eqref{eq:1} with state-dependent unknown dynamics, i.e., the uncertainty affected by control input is not included in the unknown dynamics. To broaden the applicability to a larger range of unknown systems e.g., consider an unknown $\bm{g}(\bm{x})$$\bm{u}$, the compound kernel trick \cite{valid_CBF_1} can be integrated to learn the unknown dynamics of $\bm{g}(\bm{x})\bm{u}$. However, how to sufficiently and safely excite the system for accurate predictions by choosing $\bm{u}$ in the training dataset is still an open question, which is considered for future research.
\end{remark}

\begin{remark}
    In this paper, we propose a control method to pursue the PTSf, which is guaranteed if the QP form \eqref{eq:main_theory} is feasible for the system \eqref{eq:1}. In future extensions, this feasibility assumption could be relaxed by several techniques, for example, a back-up control law \cite{lederer_safe_2023} or the online learning strategy \cite{capone2023safe}, \cite{Feasible_assump_3} can be designed to maintain feasibility.
\end{remark}

\section{Numerical Evaluation} \label{section_simulation}
In this section, we consider a two-link robotic manipulator \cite{murray2017mathematical} with unit masses and unit length for each link. Based on the robot dynamics, the state space model as \eqref{eq:1} is written as a second-order dynamics with $n = m = 2$ and
\begin{align}
	\bm{f}(\bm{x})= \bm{M}^{-1}(\bm{x}) ( -\bm{C}(\bm{x}) - \bm{G}(\bm{x}) ), && \bm{g}({\bm{x}})= \bm{M}^{-1}(\bm{x}),\label{eq:robot_dynamic} \nonumber
\end{align}
where $\bm{M}(\bm{x})$, $\bm{C}(\bm{x})$, $\bm{G}(\bm{x})$ are nominal inertia matrix, Coriolis and centrifugal term, and gravitational term from \cite{murray2017mathematical}, respectively. 
The system states $\bm{x} \!=\! [\bm{x}_1^\top, \bm{x}_2^\top]^\top$ represents joint positions and joint velocities, which are expressed as $\bm{x}_1 \!=\! [q_1, q_2]^\top \!\in\! [-2\pi,\! 2\pi]^2$ and $\bm{x}_2 \!=\! [\dot{q}_1,\! \dot{q}_2]^\top \!\in\! [-10, 10]^2$. We consider unknown dynamics in \eqref{eq:1} is $\bm{d}(\bm{x}) \!=\! [d_1(\bm{x}),\! d_2(\bm{x})]^{\top} \!=\! [5 \sin(q_1) \!+\! 3 \cos(q_2), 3 \cos(q_1) \!+\! 5\sin(q_2) \!+\! 30 ]^{\top}$. 
To identify the system uncertainty, GP regression is used with the squared exponential kernel, i.e., 
$\kappa(\bm{x},\bm{x}') \!=\! \sigma_f^2 \exp(- 0.5 l^{-2} \| \bm{x} \!-\! \bm{x}' \|^2),$
where $\sigma_f \!=\! 1$ and $l \!=\! 0.4$. The parameters of the error bound are chosen by $\delta \!=\! 0.01$ and $\tau \!=\! 10^{-10}$. For training the models, a data set $\mathbb{D}$ with $900$ data pairs are collected equally distributed on the domain $q_1, q_2 \!\in\! [-2\pi, 2\pi]$.\looseness=-1

\begin{figure}[t]
	\centering
	\includegraphics[width=0.39\textwidth]{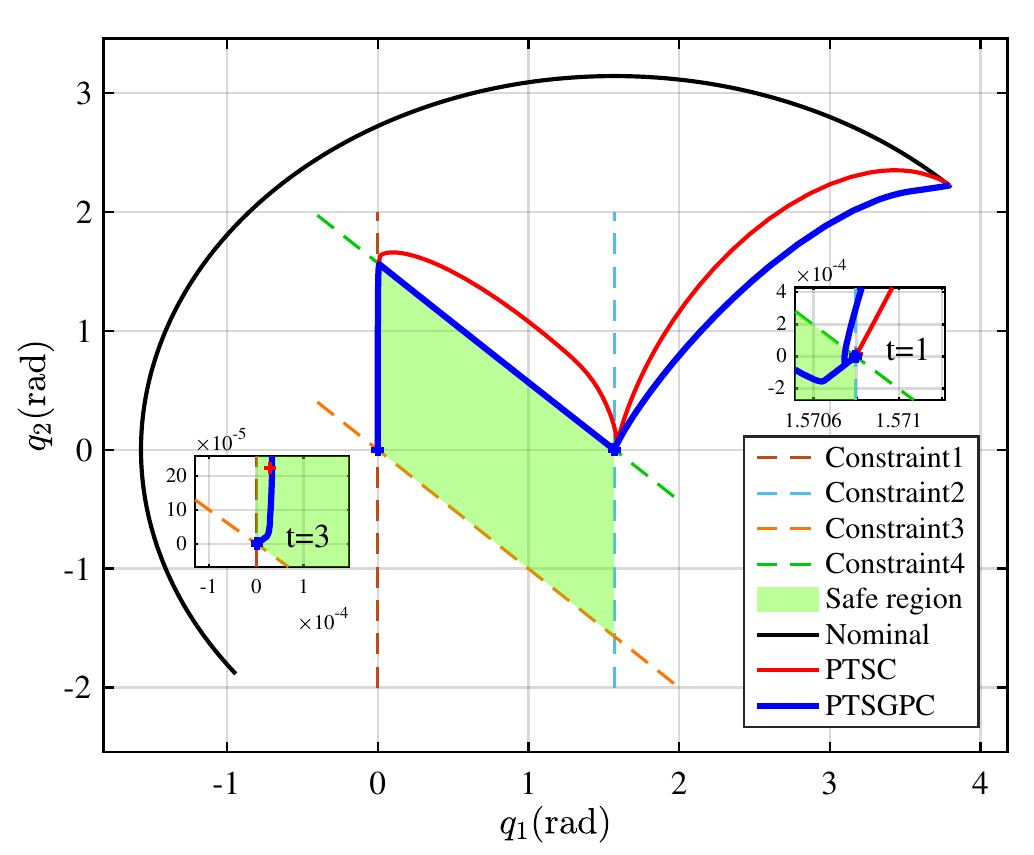}
    \vspace{-0.3cm}
	\caption{Trajectory for robot manipulator under PTSC and PTSCGP.}
    \vspace{-0.2cm}
	\label{fig:state}
\end{figure}
The safe region is defined such that the entire robot manipulator belongs to the first quadrant of Cartesian space in task space, which is equivalent to the green zone in \cref{fig:state} in joint space.
Moreover, the safe region is also expressed through the functions $h^{(i)}_1: \mathbb{R}^2 \to \mathbb{R}$ for $\forall i=1,\cdots,4$ as
\begin{align}
	&h_1^{(1)}(\bm{x}_1)=q_1, & & h_1^{(3)}(\bm{x}_1)=q_1+q_2, \nonumber \\
	&h_1^{(2)}(\bm{x}_1)=-q_1+ \pi/2, & & h_1^{(4)}(\bm{x}_1)=-q_1-q_2+\pi/2. \nonumber
\end{align}
Each function $h^{(i)}_1(\cdot)$ introduces a constraint in \eqref{eq:main_theory}, inducing a QP problem with $4$ constraints.

\begin{figure}[t]
	\centering
	\includegraphics[width=0.39\textwidth]{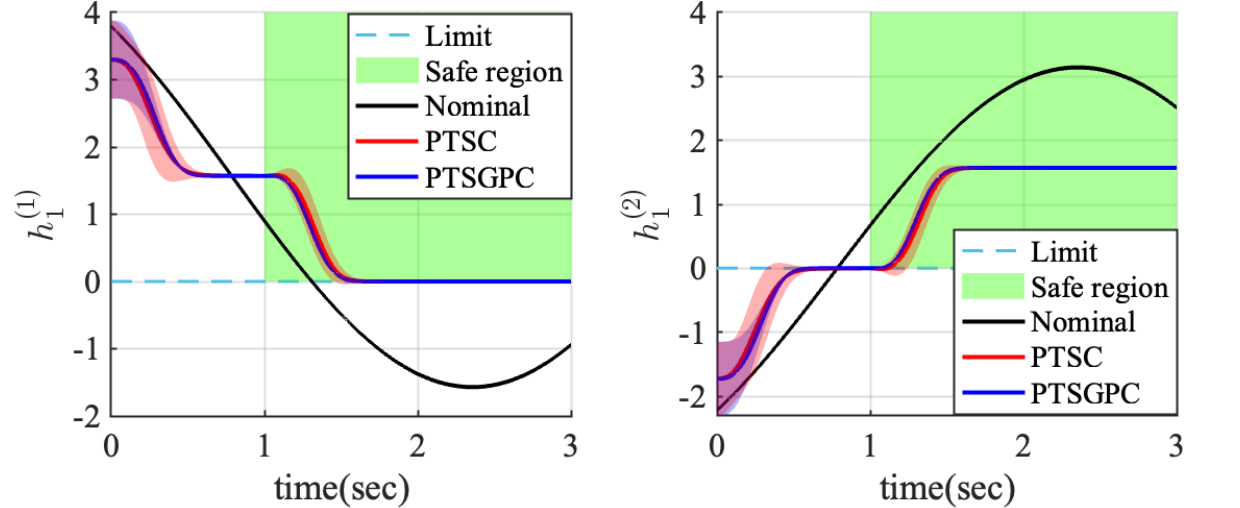}
    \includegraphics[width=0.39\textwidth]{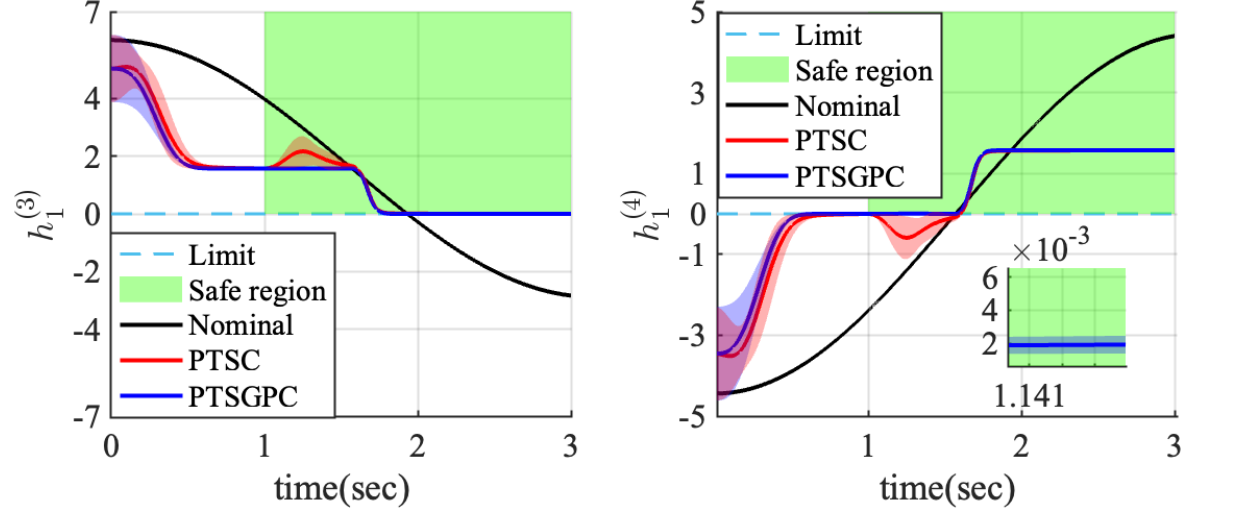}
    \vspace{-0.2cm}
	\caption{Results of $h_1^{(i)}$ with $i \!=\! 1, \!\cdots\!, 4$ for PTSC and PTSCGP in $t \!=\! [0,\!3]$.}
    \vspace{-0.6cm}
	\label{fig:h1}
\end{figure}

The nominal control task is to track a desired reference $\bm{q}_d(t) = [\pi \cos(t- \pi /4 ) + 0.5 \pi, \pi \sin ( t - \pi / 4 ) ]^{\top}$, such that the nominal control law $\bm{u}_{nom}$ is designed in PD form as $\bm{u}_{nom}=\bm{K}_p( \bm{q}_d -\bm{x}_1) + \bm{K}_d (\dot{\bm{q}}_d - \bm{x}_2)$ with $\bm{K}_p = \mathrm{diag}(50,50)$ and $\bm{K}_d=\mathrm{diag}(25,25)$.
Set the simulation time as $t \in [0, 3]$, and notably the resultant path lies outside the safe region for $t \in [0, 3]$ as shown in \cref{fig:state}. The initial state $\bm{x}(0)$ at $t_0 = 0$ is set as $\bm{x}(0) = [\bm{q}_d(0)^\top, \bm{0}_{1 \times 2}]^\top$ satisfying $\bm{x}(0) \notin \mathcal{C}$.
The safety filter is designed as in \eqref{eq:main_theory} with $\alpha=400$ and two time periods $t=[t_0^{(1)}, t_0^{(1)}+T_{\text{pre}}^{(1)}]$ as well as $t=[t_0^{(2)}, t_0^{(2)}+T_{\text{pre}}^{(2)}]$ with $t_0^{(1)}=0, T_{\text{pre}}^{(1)} = t_0^{(2)}=1, T_{\text{pre}}^{(2)}=3$. In the first time period, starting from an unsafe initial condition, the safety objective is returning to the safe region. Then, the system states maintain within the safe region during the second time period, i.e., $h_1^{(i)} \in \mathbb{R}^+$, with $i=1,\cdots,4$ for $\forall t \in [1, 3]$. In order to illustrate the validity of our proposed approach, the simulation is repeated 100 times to account for the randomness in unknown dynamics and initial states of the system, which are randomized uniformly in the range of $ \left[ \bm{d}(\bm{x})-15, \bm{d}(\bm{x})+15 \right]$ and $\left[ -1+\bm{q}_d(0), \bm{q}_d(0) \right]$, respectively.
\looseness=-1

To demonstrate the superiority of the proposed prescribed-time safe Gaussian process control (PTSGPC), the prescribed-time safe control (PTSC) proposed in \cite{abel_prescribed-time_2023} is used for comparison. The desired reference trajectory and the state trajectory of PTSC and PTSGPC are shown in \cref{fig:state}. 
The proposed PTSGPC properly addresses uncertainty, ensuring that the robot manipulator achieves the safety objective, which is also close to the nominal trajectory throughout the entire process. In contrast, although the PTSC enables the system to return to the safe region in a specified time, it fails to maintain the safety condition in the subsequent period due to the impact of uncertainty. Notably, the result of PTSC in \cref{fig:h1} shows the negative value of $h_1^{(4)}$, which violates the $4^{th}$ safety constraint, with high probability during the time period $t \in [1,3]$. Conversely, in the case of PTSGPC, all values of $h_{1}^{(i)}$ turn positive with a 95\% probability setting of $1-m\delta$ after $t=1$ attributed to the learning of uncertainty, which validates the \cref{t1} even if the system dynamics is partially known. The performance of unknown dynamics quantification is demonstrated in \cref{fig:GP_result}, which illustrates that the prediction error from GP regression is under the probabilistic error bound with 95\% probability.
The control input $\bm{u} = [\text{u1}, \text{u2}]^\top$ from the proposed control law in \cref{t1} for the robot manipulator is shown in \cref{fig:controlInput}, where $\text{u1}$ and $\text{u2}$ are the control inputs in the first and second joints, respectively.
\looseness=-1

\begin{figure}[t]
	\centering
	\includegraphics[width=0.48\textwidth]{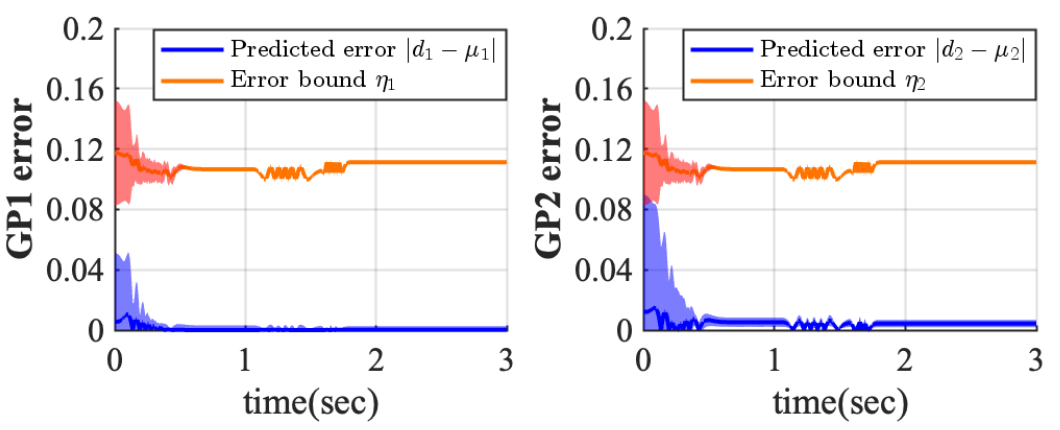}
    \vspace{-0.4cm}
	\caption{The GP prediction error of uncertainty and its error bound.}
    \vspace{-0.3cm}
	\label{fig:GP_result}
\end{figure}

\begin{figure}[t]
	\centering
	\includegraphics[width=0.35\textwidth]{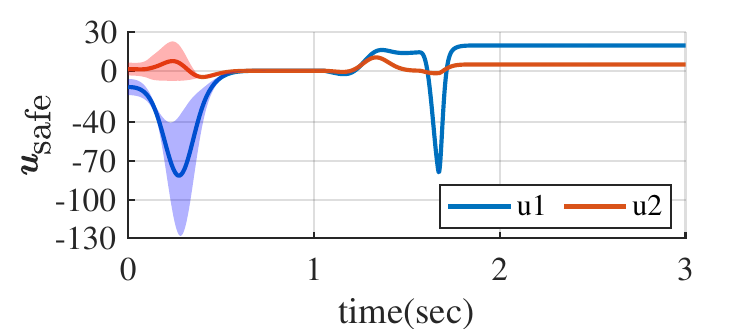}
    \vspace{-0.3cm}
	\caption{The control signals of PTSCGP.}
    \vspace{-0.6cm}
	\label{fig:controlInput}
\end{figure}

\section{Conclusions} \label{section_conclusion}

In this paper, we propose a safe learning control for control affine systems, that ensures the safety condition in a given prescribed time, independent of the initial state. By integrating the time-varying design and Gaussian process regression in the barrier function, the guarantee for system safety with a high probability is shown. The result shows that the system achieves the safety objective using our designed controller.\looseness=-1

\addtolength{\textheight}{-12cm} 

\bibliographystyle{IEEEtran}
\bibliography{ref/application,ref/CBF,ref/finite_time,ref/fixed_time,ref/Prescirbed_time,ref/NNCBF,ref/GPCBF,ref/MethodRef,ref/Supplement}

\begin{thebibliography}{10}
\providecommand{\url}[1]{#1}
\csname url@samestyle\endcsname
\providecommand{\newblock}{\relax}
\providecommand{\bibinfo}[2]{#2}
\providecommand{\BIBentrySTDinterwordspacing}{\spaceskip=0pt\relax}
\providecommand{\BIBentryALTinterwordstretchfactor}{4}
\providecommand{\BIBentryALTinterwordspacing}{\spaceskip=\fontdimen2\font plus
\BIBentryALTinterwordstretchfactor\fontdimen3\font minus
  \fontdimen4\font\relax}
\providecommand{\BIBforeignlanguage}[2]{{%
\expandafter\ifx\csname l@#1\endcsname\relax
\typeout{** WARNING: IEEEtran.bst: No hyphenation pattern has been}%
\typeout{** loaded for the language `#1'. Using the pattern for}%
\typeout{** the default language instead.}%
\else
\language=\csname l@#1\endcsname
\fi
#2}}
\providecommand{\BIBdecl}{\relax}
\BIBdecl

\bibitem{HandoverTask}
A.~Kshirsagar, G.~Hoffman, and A.~Biess, ``Evaluating guided policy search for
  human-robot handovers,'' \emph{IEEE Robotics and Automation Letters}, vol.~6,
  no.~2, pp. 3933--3940, 2021.

\bibitem{dai_adaptive_2022}
S.-L. Dai, K.~Lu, and J.~Fu, ``Adaptive finite-time tracking control of
  nonholonomic multirobot formation systems with limited field-of-view
  sensors,'' \emph{IEEE Transactions on Cybernetics}, vol.~52, no.~10, pp.
  10\,695--10\,708, 2021.

\bibitem{2020_MPC2}
J.~Berberich, J.~Köhler, M.~A. Müller, and F.~Allgöwer, ``Data-driven model
  predictive control with stability and robustness guarantees,'' \emph{IEEE
  Transactions on Automatic Control}, vol.~66, no.~4, pp. 1702--1717, 2021.

\bibitem{safe_RL1}
Y.~Zhang, M.~Chadli, and Z.~Xiang, ``Prescribed-time formation control for a
  class of multiagent systems via fuzzy reinforcement learning,'' \emph{IEEE
  Transactions on Fuzzy Systems}, vol.~31, no.~12, pp. 4195--4204, 2023.

\bibitem{ames_control_2014}
A.~D. Ames, J.~W. Grizzle, and P.~Tabuada, ``Control barrier function based
  quadratic programs with application to adaptive cruise control,'' in
  \emph{53rd IEEE Conference on Decision and Control(CDC)}, 2014, pp.
  6271--6278.

\bibitem{abel_prescribed-time_2023}
I.~Abel, D.~Steeves, M.~Krstić, and M.~Janković, ``Prescribed-time safety
  design for strict-feedback nonlinear systems,'' \emph{IEEE Transactions on
  Automatic Control}, vol.~69, no.~3, pp. 1464--1479, 2024.

\bibitem{li_formally_2018}
A.~Li, L.~Wang, P.~Pierpaoli, and M.~Egerstedt, ``Formally correct composition
  of coordinated behaviors using control barrier certificates,'' in \emph{2018
  IEEE/RSJ International Conference on Intelligent Robots and Systems (IROS)},
  2018, pp. 3723--3729.

\bibitem{garg_multi-rate_2022}
K.~Garg, R.~K. Cosner, U.~Rosolia, A.~D. Ames, and D.~Panagou, ``Multi-rate
  control design under input constraints via fixed-time barrier functions,''
  \emph{IEEE Control Systems Letters}, vol.~6, pp. 608--613, 2022.

\bibitem{taylor_learning_nodate}
A.~Taylor, A.~Singletary, Y.~Yue, and A.~Ames, ``Learning for safety-critical
  control with control barrier functions,'' in \emph{Learning for Dynamics and
  Control}.\hskip 1em plus 0.5em minus 0.4em\relax PMLR, 2020, pp. 708--717.

\bibitem{yaghoubi_training_2020}
S.~Yaghoubi, G.~Fainekos, and S.~Sankaranarayanan, ``Training neural network
  controllers using control barrier functions in the presence of
  disturbances,'' in \emph{2020 IEEE 23rd International Conference on
  Intelligent Transportation Systems (ITSC)}, 2020, pp. 1--6.

\bibitem{jagtap_control_2020}
P.~Jagtap, G.~J. Pappas, and M.~Zamani, ``Control barrier functions for unknown
  nonlinear systems using {{G}}aussian processes,'' in \emph{59th IEEE
  Conference on Decision and Control (CDC)}, 2020, pp. 3699--3704.

\bibitem{lederer_safe_2023}
A.~Lederer, A.~Begzadi{\'c}, N.~Das, and S.~Hirche, ``Safe learning-based
  control of elastic joint robots via control barrier functions,''
  \emph{IFAC-PapersOnLine}, vol.~56, no.~2, pp. 2250--2256, 2023.

\bibitem{zhang2023safety}
S.~Zhang, D.-H. Zhai, Y.~Xiong, J.~Lin, and Y.~Xia, ``Safety-critical control
  for robotic systems with uncertain model via control barrier function,''
  \emph{International Journal of Robust and Nonlinear Control}, vol.~33, no.~6,
  pp. 3661--3676, 2023.

\bibitem{huang2023rapid}
H.-L. Huang, M.-Y. Cheng, and T.-Y. Huang, ``A rapid base parameter physical
  feasibility test algorithm for industrial robot manipulator identification
  using a recurrent neural network,'' \emph{IEEE Access}, vol.~11, pp.
  145\,692--145\,705, 2023.

\bibitem{slotine1993robust}
J.-J.~E. Slotine and J.~Karl~Hedrick, ``Robust input-output feedback
  linearization,'' \emph{International Journal of control}, vol.~57, no.~5, pp.
  1133--1139, 1993.

\bibitem{capone2019backstepping}
A.~Capone and S.~Hirche, ``Backstepping for partially unknown nonlinear systems
  using {{G}}aussian processes,'' \emph{IEEE Control Systems Letters}, vol.~3,
  no.~2, pp. 416--421, 2019.

\bibitem{hashimoto2022learning}
K.~Hashimoto, A.~Saoud, M.~Kishida, T.~Ushio, and D.~V. Dimarogonas,
  ``Learning-based symbolic abstractions for nonlinear control systems,''
  \emph{Automatica}, vol. 146, p. 110646, 2022.

\bibitem{Aron_2017}
A.~D. Ames, X.~Xu, J.~W. Grizzle, and P.~Tabuada, ``Control barrier function
  based quadratic programs for safety critical systems,'' \emph{IEEE
  Transactions on Automatic Control}, vol.~62, no.~8, pp. 3861--3876, 2017.

\bibitem{castaneda2022probabilistic}
F.~Castaneda, J.~J. Choi, W.~Jung, B.~Zhang, C.~J. Tomlin, and K.~Sreenath,
  ``Probabilistic safe online learning with control barrier functions,''
  \emph{arXiv preprint arXiv:2208.10733}, 2022.

\bibitem{lederer_uniform_nodate}
A.~Lederer, J.~Umlauft, and S.~Hirche, ``Uniform error bounds for {{G}}aussian
  process regression with application to safe control,'' \emph{Advances in
  Neural Information Processing Systems}, vol.~32, 2019.

\bibitem{venkatesh2013theory}
S.~S. Venkatesh, \emph{The theory of probability: Explorations and
  applications}.\hskip 1em plus 0.5em minus 0.4em\relax Cambridge University
  Press, 2013.

\bibitem{Union_bound}
N.~Srinivas, A.~Krause, S.~M. Kakade, and M.~W. Seeger, ``Information-theoretic
  regret bounds for {{Gaussian}} process optimization in the bandit setting,''
  \emph{IEEE Transactions on Information Theory}, vol.~58, no.~5, pp.
  3250--3265, 2012.

\bibitem{hassan2002nonlinear}
K.~K. Hassan \emph{et~al.}, ``Nonlinear systems,'' \emph{Departement of
  Electrical and computer Engineering, Michigan State University}, 2002.

\bibitem{lederer2021gaussian}
A.~Lederer, A.~J.~O. Conejo, K.~A. Maier, W.~Xiao, J.~Umlauft, and S.~Hirche,
  ``Gaussian process-based real-time learning for safety critical
  applications,'' in \emph{International Conference on Machine Learning}.\hskip
  1em plus 0.5em minus 0.4em\relax PMLR, 2021, pp. 6055--6064.

\bibitem{capone2023safe}
A.~Capone, R.~Cosner, A.~Ames, and S.~Hirche, ``Safe online dynamics learning
  with initially unknown models and infeasible safety certificates,''
  \emph{arXiv preprint arXiv:2311.02133}, 2023.

\bibitem{valid_CBF_1}
F.~Castaneda, J.~J. Choi, B.~Zhang, C.~J. Tomlin, and K.~Sreenath, ``Pointwise
  feasibility of {{Gaussian}} process-based safety-critical control under model
  uncertainty,'' in \emph{2021 60th IEEE Conference on Decision and Control
  (CDC)}.\hskip 1em plus 0.5em minus 0.4em\relax IEEE, 2021, pp. 6762--6769.

\bibitem{Feasible_assump_3}
F.~Castaneda, J.~J. Choi, W.~Jung, B.~Zhang, C.~J. Tomlin, and K.~Sreenath,
  ``Probabilistic safe online learning with control barrier functions,''
  \emph{arXiv preprint arXiv:2208.10733}, 2022.

\bibitem{murray2017mathematical}
R.~M. Murray, Z.~Li, and S.~S. Sastry, \emph{A mathematical introduction to
  robotic manipulation}.\hskip 1em plus 0.5em minus 0.4em\relax CRC press,
  2017.

\end{thebibliography}

\end{document}